\documentclass[final,leqno,onefignum,onetabnum]{siamltexmm}
\usepackage{amssymb,amsmath}
\usepackage{empheq}
\usepackage{graphicx}
\usepackage{color}
\usepackage{soul}

\title{COMPLEXITY REDUCTION IN MANY PARTICLE SYSTEMS WITH RANDOM INITIAL DATA
\thanks{The 
        work of Leonid Berlyand and Mykhailo Potomkin was supported by DOE grant DE-FG-0208ER25862. The work of Pierre-Emmanuel Jabin was partially supported by NSF grant DMS-1312142. LB and MP wish to thank V. Rybalko for his comments and suggestions which helped to improve the manuscript.}} 

\author{Leonid Berlyand\thanks{Department of Mathematics, The Pennsylvania State University, University Park, Pennsylvania. 
(\email{lvb2@psu.edu})}
\and 
Pierre-Emmanuel Jabin\thanks{Department of Mathematics, University of Maryland, College Park, MD 20742
USA, (\email{pjabin@umd.edu})}
\and Mykhailo Potomkin\thanks{Department of Mathematics, The Pennsylvania State University, University Park, Pennsylvania, (\email{potomkin@math.psu.edu})}
}

\begin{document}
\maketitle
\newcommand{\slugmaster}{%
\slugger{}{xxxx}{xx}{x}{x--x}}

\begin{abstract}
We consider the motion of interacting particles governed by a coupled system of ODEs with random initial conditions. Direct computations for such systems are prohibitively expensive due to a very large number of particles and randomness requiring many realizations in their locations in the presence of strong interactions. While there are several approaches that address the above difficulties, none addresses all three simultaneously. Our goal is to develop such a computational approach in order to capture the experimentally observed emergence of correlations in the collective state (patterns due to strong interactions). Our approach is based on the truncation of the BBGKY hierarchy that allows one to go beyond the classical Mean Field limit and capture correlations while drastically reducing the computational complexity. 
Finally, we provide an example showing a numerical solution of this nonlinear and non-local
system.
\end{abstract}



\pagestyle{myheadings}
\thispagestyle{plain}
\markboth{L. Berlyand, P.-E. Jabin, M. Potomkin }{Many Particle Systems with Randomness}

\section{Motivation and Settings}

Systems of interacting particles described by a coupled system of a large
number of ODEs with random initial conditions appear in many problems of physics, cosmology,
chemistry, biology, social science and economics:
\begin{equation}\label{original}
\dot{X}_i=S(X_i)+\frac{\alpha}{N}\sum\limits_{j=1}^{N}K(X_j-X_i),\;\;i=1,...,N.
\end{equation}
Here $X_i(t)$ denotes the position of $i$th particle and $X_i$ belongs to $D$, where $D$ throughout this paper can stand for ${\mathbb R}^d$, a $d$-dimensional torus $\Pi^d$, or a compact domain in ${\mathbb R}^{d}$ in which case boundary conditions must be added. The scalar function $K$ describes the inter-particle interactions, and $S(X_i)$ models an internal force of each particle, such as self-propulsion. 

System \eqref{original} is an Individual Based Model, i.e., it has an ODE for each particle coupled with others. In various applications the role of an individual can be played by atoms, bacteria in suspensions ({microswimmers}), animals in flocks, social agents etc. The system \eqref{original} has two key parameters: $\alpha$, the strength of interactions, and $N$, the number of particles.  
The parameter $\alpha$ is determined by both geometry such as interparticle distance and the {mass} of a particle (note that a model particle is just a point) as well as physics. 
In this paper we restrict ourselves to  the case when the right hand side of \eqref{original} is linear in $\alpha$. The magnitude of $\alpha$ plays an important role in analysis of the system \eqref{original}: a small $\alpha$ corresponds to almost decoupled {interactions}; large $\alpha$ corresponds to strong interactions which is our main focus; $\alpha \sim 1$ corresponds to the classical Mean Field regime.  

Our work is motivated by experiments in bacterial suspensions \cite{WuLib00,DomCisCha04,SokAra07,SokAra09,SokApoGrz10,LepKesGan11,SokAra12}. These experiments \cite{SokAra07,SokAra09,SokAra12} show the emergence of a coarse {\it collective scale} 
 when the concentration of bacteria
exceeds a critical value. Roughly speaking, the collective scale is the correlation length of the velocity field
 in a bacterial suspension. A striking universality property has been observed experimentally and numerically in \cite{SokAra07,SokAra09,RyaSok13} the collective scale does not change when swimming speed and concentration increase, that is more energy is injected into the system (for other studies of collective state in bacterial suspensions see also \cite{SanShe13} and references therein).

The motion of bacteria can be modeled by a system of the form \eqref{original} where the position and orientation of the $i$th bacteria are described by the vector $X_i(t)$. In this case  the parameter $\alpha$ equals  $\left({\ell}/{R}\right)^2 NV_0$, where $V_0$  is \textcolor{black}{the swimming speed of a single bacterium}, $R$ is the mean distance between two bacteria, and $\ell$ is the characteristic size of a bacterium. 
The collective behavior observed in experiments \cite{SokAra07,SokAra09,SokAra12} has also been qualitatively reproduced by direct numerical simulations in \cite{RyaSok13} for systems of $10^5$ bacteria, which validates the model of the type \eqref{original}. 
However, the computational cost of direct simulations of the ODE system for a realistic number of bacteria is prohibitively high for the following reasons: 
\begin{list}{}{}
\item{\hspace{-0.3in}$(i)$} \hspace{4 pt} the number of bacteria $N$ is very large ($10^{10}$ per $\text{cm}^3$);
\item{\hspace{-0.33in}$(ii)$} \hspace{6pt}to draw a reliable conclusion one needs to consider many realizations, which mathematically translates into random initial data;
\item{\hspace{-0.35in}$(iii)$} \hspace{3 pt}the main interest is in collective state corresponding to large $\alpha$ which leads to small time steps. 
\end{list}
The combination of the factors $(i)$-$(iii)$ makes the computational cost too high even for the most powerful particle methods such as Fast Multipole Method \cite{GreRok87,GreRok88,GreRok89}. 

The goal of this paper is to propose a computational approach that allows one to describe numerically the collective state of this system with properties $(i)$-$(iii)$. More specifically, the collective state is described by the correlation length and two-point correlation function. The objective of our study is the efficient computation of these two quantities. 

The main idea is to replace the ODE system \eqref{original} by a PDE such that the computational cost of its solution does not grow as $N$ goes to infinity. This idea had been used in the classical Mean Field 
approach which corresponds to $\alpha$ of the order 1 and is not valid for strong interactions, e.g., $\alpha \sim N^{\gamma}$ for $0<\gamma<1$.

This paper focuses on a PDE approach that extends beyond the Mean Field, so that it can capture correlations in a computationally efficient way such that the computational complexity {increases only slowly} with $N$. The main idea is to consider the BBGKY hierarchy of PDEs which consists of $N$ equations (therefore it is even harder to solve than \eqref{original}) and obtain a closed system for 1-particle and 2-particle distributions by a clever truncation of the hierarchy. Then the large parameter $N$ is present only in coefficients in a more innocuous way, and they can be handled efficiently with high order methods. This approach computes distribution functions and therefore it avoids computing individual realizations. 
Thus, it allows us to overcome the computational difficulties $(i)$ and $(ii)$. The contribution to the computational complexity from difficulty $(iii)$ is still present but much less of a problem than $(i)$ and $(ii)$, because $\alpha \sim N^{\gamma}$ and $\gamma<1$.

Note that a specific feature of our method is that it is efficient for random initial conditions of system \eqref{original} in contrast to deterministic. Indeed, a seemingly simpler case of deterministic initial data leads to a solution of the truncated BBGKY hierarchy with singular initial conditions ($\delta$-functions), which is why the numerical cost of solving such a deterministic problem is very high.  In contrast, random initial data in ODE \eqref{original} lead to {\it smooth} initial conditions in the truncated BBGKY hierarchy that is much easier to handle numerically. 

The truncation presented in this paper can be applied for various ODE systems of type \eqref{original}, {in dimension $1$ or more}. Note that different truncations of the Boltzmann hierarchy have been made before for some specific situations, which usually rely on some perturbative arguments. For example, we refer to papers \cite{Mar87,HonNieOtt05,HonNieOtt06} devoted to Ostwald ripening where a truncation was motivated by expansions in (small) concentration  of particles.   

The paper is organized as follows. We recall the Mean Field approach and discuss its limitations in Section \ref{meanfield}. The truncation of BBGKY hierarchy is described in Section 3. Numerical simulations performed to check that the truncated PDE system is reliable are decribed in Section \ref{numerics}.

\section{Random initial conditions, correlations, the Mean Field approach}\label{meanfield}

For physical reasons, the initial conditions for the system \eqref{original} are typically random as explained below. In the classical Mean Field theory, this leads to a drastic reduction in the computational complexity: it is possible to approximate the original solution by the solution of a PDE which does not depend on $N$. We describe here the two classical approaches  to derive the Mean Field limit. The first one is based on the so-called empirical measure. The second one is a statistical approach which is better suited  for our purpose. 

The Mean Field limit is valid as long as the correlations between particles are negligible. This phenomenon is known as {\it propagation of chaos}. However, our work is motivated by experimental studies of the collective state, whose key feature is the rise of {\it correlations} corresponding to the emergence of a {\it collective scale}. In this case, as we will explain below, the Mean Field approach fails.   

Our approach in this paper is mostly formal. Nevertheless, we point out that the Mean Field theory described below can be made rigorous if some smoothness is assumed on $K$. More precisely, 
\begin{equation}
\nabla K\in L^\infty(D),\qquad K(x)\rightarrow 0\quad\mbox{as}\quad |x|\rightarrow \infty. \label{condK} 
\end{equation}
On the other hand, we believe that the numerical implementation of this approach will work well even for singular kernels (see remark 2.1).

\medskip 

\subsection{Preliminaries}
\noindent{\it How to choose initial conditions}: {\it Randomness and marginals.} 
By assumption \eqref{condK} and the standard Cauchy-Lipschitz theory, there exists a unique solution to \eqref{original} once each initial position $X_i(0)$ is chosen. However for most practical purposes, determining those initial positions can be a very delicate problem  as the full information is not accessible from an experimental point of view. For $N\sim 10^{10}$, it is indeed completely unrealistic to measure the position of each particle with enough precision.

Instead, {what is accessible is some} statistical information about the positions of the particles. Hence one  usually assumes that the initial position of each particle is randomly distributed. That means that the information on the initial distribution of the particles is now encoded in the $N$-particle distribution function at time $0$, $f_N(t=0,x_1,...,x_N)$. Given a subdomain $\mathcal{Q}\subset D^N$, the probability of finding the initial positions $(X_1^0,...,X_N^0)\in\mathcal{Q}$ is given by
\[
\int_{\mathcal{Q}} f_N(t=0,x_1,...,x_N)\,dx_1...dx_N.
\]  
System \eqref{original} is deterministic but if the initial conditions are random, then the randomness will be propagated  defining the $N$-particle distribution for $t>0$. Technically, $f_N(t,.)$ is the push forward of $f_N(t=0,.)$ by the flow generated by \eqref{original}. 

From $f_N$ one may define the $k$-th marginal
\[
f_k(t,x_1,...,x_k)=\int_{D^{N-k}} f_N(t,x_1,...,x_k,x_{k+1},...,x_N)\,dx_{k+1}...dx_N.
\]
Some of marginals have a natural physical interpretation. For instance, $f_1$ is the $1$-particle distribution function and for $\mathcal{O}\subset D$ the average number of particles in the subset $\mathcal{O}$ is
\[
\int_{\mathcal{O}} f_1(t,x)\,dx.
\]

It is still not experimentally possible to measure $f_N$ but it is possible to measure some marginals, especially $f_1$ and the $2$-particle distribution function $f_2$. 

In the simplest case, one assumes that the particles are initially independently and identically 
distributed, that is
\begin{equation}
f_N(t=0,x_1,...,x_N)=\Pi_{i=1}^N f^0(x_i).\label{initiallaw}
\end{equation}
This independence is strongly connected to the usual Mean Field limit approach as explained in subsection 2.2 (see \eqref{indep}).

\medskip

\noindent{\it Definition of correlations.} 
Our main goal is to understand how correlations develop in system \eqref{original} with random initial conditions. Those are connected to the second marginal $f_2$.

We define correlation  of particles' positions by 
\begin{equation}\label{def_corr}
c= \frac{{\mathbb E}[X_1\cdot X_2]-\left({\mathbb E}[X]\right)^2 }{{\mathbb E}[X^2]-\left({\mathbb E}[X]\right)^2}=\frac{\int x_1\cdot x_2 f_2(x_1,x_2)dx_1dx_2-\left(\int xf_1(x) dx\right)^2}{\int x^2 f_1(x)dx-\left(\int xf_1(x) dx\right)^2}.
\end{equation}

Observe that the correlation $c$ can only vanish if
\[
f_2(x,y)=f_1(x)\,f_1(y),
\]
that is if the particles positions are independent. Therefore, roughly speaking, the correlations in the system are determined by how far $f_2(x,y)$ is from $f_1(x)\,f_1(y)$.

\medskip

\subsection{The Mean Field approach} $\;$\\
\noindent{\it Empirical measure.} Assume that the $X_i(t)$ are solutions to \eqref{original}, and define the empirical measure
\begin{equation}\label{empirical}
\mu_N(t,x)=\frac{1}{N}\sum\limits_{i=1}^{N} \delta (x-X_i(t)). 
\end{equation}
Note that if the particles are undistinguishable then there is just as much information in the empirical measure as in the position vector $(X_1,...,X_N)$. Otherwise, it only tells that there is a particle at $x$, but it is not clear which one.  

If $K$ is continuous, then $\mu_N$ solves the Vlasov equation in the sense of distributions
\begin{equation}
\partial_t f(t,x)+\nabla_{x}\cdot \left(S(x)f(t,x)\right)+\alpha\nabla_{x}\cdot\left(\int K(y-x)f(t,y)dy f(t,x)\right)=0.
\label{vlasov}
\end{equation}

Consider a sequence of initial positions $\mathcal{X}_N=\left\{X_{i}(0):i=1,...,N\right\}$ such that the corresponding empirical measure
$\mu_N(0)$ converges to some $f^0\in { \Pi}(D)$ as $N$ goes to infinity. Here ${\Pi}(D)$ is the space of probability measures $\mu$ on $D$ such that $\mu(D)=1$. Then it is natural to expect that $\mu_N$ will also converge to the corresponding solution $f$ to \eqref{vlasov} with initial data $f^0$. 
Assuming that $f^0$ is smooth then it is possible to compute numerically $f$ and hence to get a good approximation to $\mu_N$. This is the classical Mean Field limit theory which can be made quantitative.

Those quantitative estimates require some weak distances on the space of measures. These are classically the so-called Monge-Kantorovich-Wasserstein (MKW) distances. For our purpose it is enough to understand that they {correspond to} some appropriate distance between probability measures. For the sake of completeness we define these distances in Appendix A.

Now we give the main stability estimate behind the Mean Field limit.
From \cite{BraHep77}, \cite{Neun79}, and \cite{Spoh91}, it is possible to prove
that if $f$ and $g$ are two measure-valued solutions to \eqref{vlasov}, then
\begin{equation}
W_p(f(t,.),g(t,.))\leq e^{t\,\alpha\,\|\nabla K\|_{L^\infty}}\,W_p(f(0,.),g(0,.)),
\label{gronwall}
\end{equation}
where $W_p(\cdot,\cdot)$ is a $p$-Wasserstein or MKW distance between two measures. The inequality \eqref{gronwall} is a Gronwall-type inequality. Note also that the inequality \eqref{gronwall} applies for any initial conditions $f(0,.)$ and $g(0,.)$ which are not necessarily random.

\medskip

\noindent{\it The Mean Field limit.} 
In our context, the initial conditions are random as it was explained before. In particular, the empirical measure at time $t=0$ is itself random.

If the initial law is chosen according to \eqref{initiallaw}, then a large deviation for the law of large numbers applies  (${\mathbb E}\mu_N=f^0$) and ensures that, in fact, the initial measure $\mu_N(t=0)$ is very close to $f^0$. More precisely, it is proved for example in Boissard \cite[Appendix A, Proposition 1.2]{BoissardPhD},  that if $f^0$ is a nonnegative measure with compact support of diameter $R$, then for some constant $C$ and positive coefficients $\gamma_1$ and $\gamma_2$, depending only on the dimension of $D$ and $R$
\begin{equation} \label{W1largedev}
{\mathbb P} \left( W_1(\mu_N(t=0),f^0)  \geq   \frac {C\, R} {N^{\gamma_1}}  \right) \leq
e^{-C N^{\gamma_2}}.
\end{equation}
This says that with {a probability exponentially close} to $1$, $\mu_N(t=0)$ and $f^0$ are polynomially close in $N$. Denote by $f(x,t)$ the solution to \eqref{vlasov} with $f^0(x)$ as an initial data. By combining the deterministic stability \eqref{gronwall} with a law of large numbers in the form {of} \eqref{W1largedev} we obtain that with  probability larger than $(1-e^{-C N^{\gamma_2}})$
\begin{equation}\label{stab_random}
W_1(\mu_N(t,.),\;f(t,.))\leq \frac {C\, R} {N^{\gamma_1}}\,e^{t\,\alpha\,\|\nabla K\|_{L^\infty}}.
\end{equation}

\noindent{\it Why random initial conditions make computations much easier in the Mean Field framework?}
Looking for a  solution of the Vlasov equation \eqref{vlasov} in the form of a sum of $N$ Dirac masses like $\mu_N$ is just as complex and computationally costly as solving the original  ODE system \eqref{original}. 

However looking for smooth solutions to the Vlasov equation \eqref{vlasov} is comparatively much faster and obviously independent of $N$ (provided the solution of \eqref{vlasov} is independent of $N$).  Since  the initial distribution $f^0$ is usually assumed to be smooth,  the corresponding solution $f$ is  smooth as well.  Computing   $f$ numerically is thus far easier than solving  \eqref{original}, because computational cost does not depend on $N$.

The key to the reduction in the computational complexity in this Mean Field approach is that one does not solve the original ODE system \eqref{original} but instead one solves the Vlasov PDE for $f$. The previous inequality \eqref{stab_random} implies that this  $f$ will be a good approximation of the original $\mu_N$ up to a time $t$ of order
\begin{equation}\label{time_interval}
\frac{\log N}{\alpha\,\|\nabla K\|_{L^\infty}}.
\end{equation}
Note that in certain circumstances, this time can be considerably extended to become polynomial in $N$. This usually requires a stable equilibrium to equation \eqref{vlasov}, see \cite{CagRou07} for instance.

\medskip

\subsection{Propagation of chaos} It is possible to interpret the Mean Field limit in terms of the propagation of chaos on the marginals. For this, we introduce the hierarchy of equations on the marginals.

First, it is well-known that $f_N$ solves the Liouville equation:
\begin{eqnarray}
&&\partial_t f_N +\sum\limits_{i=1}^{N}\partial_{x_i}\left(S(x_i)f_N\right)+\frac{\alpha}{N}\sum\limits_{i=1}^{N}\partial_{x_i}\left(\sum\limits_{j=1}^{N}K(x_j-x_i)f_N\right)=0.\label{eq_for_fN}
\end{eqnarray}
By integrating the equation for $f_N$, one obtains an equation satisfied by each marginal $f_k$
\begin{eqnarray}
\nonumber&&\partial_t f_k+\sum\limits_{i=1}^{k}\partial_{x_i}\left(S(x_i)f_k\right)+\frac{\alpha}{N}\sum_{i=1}^k \sum_{j=1}^k \partial_{x_i}\left(K(x_j-x_i)\,f_k\right)\\&&\hspace{60 pt}+\frac{\alpha (N-k)}{N}\sum_{i=1}^k
\partial_{x_i}\left( \int K(y-x_i)\,f_{k+1}(t,x_1,\ldots,x_k,y)\,dy\right)=0.\label{eq_for_fk}
\end{eqnarray}
For example, the PDE for $f_1$ is 
\begin{eqnarray}
&&\nonumber \partial_t f_1(t,x_1)+\partial_{x_1}\left(S(x_1)f_1(t,x_1)\right)+\frac{\alpha K(0)}{N}\partial_{x_1}f_1(t,x_1)\\&&\hspace{80 pt}+\alpha\frac{N-1}{N}\partial_{x_1}\left\{\int K(y-x_1)f_2(t,x_1,y)dy\right\}=0.\label{eq_for_f1}
\end{eqnarray}
By taking the formal limit $N\to \infty$ in the equation \eqref{eq_for_f1} and assuming the independence condition $f_2(t,x_1,x_2)=f_1(t,x_1)f_1(t,x_2)$, we get equation \eqref{vlasov}.

This leads us to the crucial concept of {\it propagation of chaos}.  Under some mild conditions on the smoothness of $K$, for  
initial positions that are close to being independent (that is \eqref{initiallaw} is assumed)  as $N\rightarrow \infty$ we have 
\begin{equation}\label{indep}
f_k(t,x_1,...,x_k)\to \Pi_{i=1}^{k}f(t,x_i),
\end{equation}
where $f(t,x)$ solves the Mean Field equation \eqref{vlasov}. 

Note that for a finite $N$, one cannot have equality in \eqref{indep} and, in particular, 
$\Pi_{i=1}^{N}f(t,x_i)$ cannot be a solution to \eqref{eq_for_fN}. Hence, for a finite but large $N$ and for initial conditions of the form $f_N|_{t=0}=\Pi_{i=1}^{N}f_0(x_i)$, the particles' positions are not independent  but their correlation  is very small, at least on the time interval when the Mean Field limit holds, $i.e.$, up to a time of order \eqref{time_interval}.

\medskip

\noindent{\it Beyond Mean Field: The BBGKY hierarchy and its truncation.} The Mean Field limit leads to a closed equation on $f_1$ but it only offers a rough estimate of $f_2$. In particular it cannot give any estimate on correlations since it relies on the premises that they are vanishing. In our context, however, this means that we cannot use the Mean Field framework to evaluate correlations as defined by \eqref{def_corr} which are small but non $0$ either, in line with the experimentally observed phenomenon that we wish to explain. 


In general the exact computation of those correlations would require one to exactly solve equation on $f_2$. As it was pointed out above, this equation would in turn require to solve the equation on $f_3$ and so on. 

Any exact solution would require solving the full equation \eqref{eq_for_fN} on $f_N$. Unfortunately, $f_N$ is a function of $N+1$ variables and the computational cost of the numerical solution of \eqref{eq_for_fN} is much too large to be even remotely realistic.

We would like instead to compute directly the marginals up to $f_k$ for some $k$, approximately if it is not possible to do it exactly.
This leads us to the key question of possible {\em truncations for the BBGKY hierarchy}. A truncation at level $k$ is an ansatz which expresses the terms involving $f_{k+1}$ in terms of $f_k$ and lower order marginals. Using this ansatz makes the first $k$ equations of the hierachy closed thus letting us solve them.

In that sense, the Mean Field limit can be seen as a particular case of truncation at order $k=1$. In this paper, we focus and propose a possible truncation at order $k=2$. We are able to show through numerical experiments that it is valid as long as correlations are not too large.


\noindent{{\bf Remark 2.1}} ({\it On singular kernels}) As mentioned before, the rigorous justification of the classical Mean Field theory requires some smoothness on the interaction kernel, $K(x)$ is Lipschitz. Many physical kernels are more singular, in particular in the context we are interested in, $i.e.,$ the context of bacteria interacting through a fluid. 

It is widely conjectured that the Mean Field theory can be extended to more singular kernels and some results are already available, see for example \cite{GooHouLow90}, \cite{Hau09}, \cite{Scho96} or \cite{HauJab07} in the phase space framework.

In this work, we are not concerned with rigorous justification of our results under proper assumptions on smoothness of $K$, however, just as in the Mean Field approach, we believe that the numerical implementation of our approach will work well for a wide class of kernels $K$ (including singular ones).  


\section{Truncation of the hierarchy} \label{truncation}
In this section we first discuss the possibility of a truncation ansatz $f_3=\mathcal{F}[f_1,f_2]$ such that the full BBGKY hierarchy becomes a system of two PDEs for marginals $f_1$ and $f_2$ only. A truncation ansatz of the form $f_3=\mathcal{F}[f_1,f_2]$ changes the equation for $f_2$:  
\begin{eqnarray}
\nonumber&&\partial_t f_2 +\frac{\alpha K(0)}{N}\partial_{x_1}f_2 +\frac{\alpha K(0)}{N}\partial_{x_2}f_2\\
\nonumber&&\hspace{50 pt}+ \frac{\alpha}{N}\partial_{x_1} (K(x_2-x_1)f_2)+\frac{\alpha}{N}\partial_{x_2}(K(x_1-x_2)f_2)\\
\nonumber&&\hspace{50 pt} +\alpha\frac{N-2}{N}\partial_{x_1}\left\{\int K(x_3-x_1)\mathcal{F}[f_1,f_2](t,x_1,x_2,x_3)dx_3\right\}\\
\label{eq_for_f2} &&\hspace{50 pt} +\alpha\frac{N-2}{N}\partial_{x_2} \left\{\int K(x_3-x_2)\mathcal{F}[f_1,f_2](t,x_1,x_2,x_3)dx_3\right\}=0.
\end{eqnarray}
For the sake of simplicity, in this section we consider the case with no self-interactions, that is $S(x)~\equiv~0$. We want the solution  to satisfy the following properties:  
\begin{enumerate}
\item $f_2(x_1,x_2)=f_2(x_2,x_1)$ ({\it particles are identical});
\item $\int\int f_2dx_1dx_2 \equiv \text{const}$, $f_1,f_2\geq 0$ provided that initial conditions for $f_1$ and $f_2$ are positive ({\it mass preserving and positivity});
\item $f_1=\int f_2$ ({\it consistency});
\item If $f_2=f_1\otimes f_1$, then $f_3(x_1,x_2,x_3)=\mathcal{F}[f_1,f_1\otimes f_1]=f_1(x_1)f_1(x_2)f_1(x_3)$.
\end{enumerate}
 By $f_2=f_1\otimes f_1$ we mean the equality $f_2(x,y)=f_1(x)f_1(y)$. 

Property 4 guarantees that the truncation ansatz $f_3=\mathcal{F}[f_1,f_2]$ is compatible with the Mean Field limit as $N \to \infty$. More precisely, letting $N\to \infty$ in the equation \eqref{eq_for_f2} for $f_2$ with the ansatz, the equation 
\begin{eqnarray}
\nonumber&&\partial_t f^{\infty}_2 +\partial_{x_1}\left\{\int K(x_3-x_1)\mathcal{F}[f^{\infty}_1,f^{\infty}_2](t,x_1,x_2,x_3)dx_3\right\}\\\nonumber &&\hspace{50 pt}+\partial_{x_2}\left\{\int K(x_3-x_2)\mathcal{F}[f_1^{\infty},f_2^{\infty}](t,x_1,x_2,x_3)dx_3\right\}=0
\end{eqnarray}
reduces to the Vlasov equation for the Mean Field limit, because the propagation of chaos holds: $f_2^{\infty}(t,x_1,x_2)=f_1^{\infty}(t,x_1)f_1^{\infty}(t,x_2)$. 

\medskip 

We reformulate  these properties as requirements on the function $\mathcal{F}$ and then prove that such an ansatz does not exist. Next, we present a truncation which is not based on a unique ansatz, yet the corresponding solution $f_2$ satisfies {the} four properties above. 
   
Consider a representation for $f_3$:
\begin{equation}\label{repr}
f_3(x_1,x_2,x_3)=\mathcal{F}[f_1,f_2](x_1,x_2,x_3),
\end{equation}
where $\mathcal{F}$ is a function  ({in general,} a nonlinear operator) of $f_1$ and $f_2$. We reformulate the key properties as requirements on $f_3$ calculated by \eqref{repr} for given $f_1$ and $f_2$. 

First, the symmetry of $f_2$ with respect to arguments $x_1$ and $x_2$ is {\it equivalent to}:    
\begin{eqnarray}
\nonumber&&\text{$f_2(x_1,x_2)=f_2(x_2,x_1)$ for all $x_1,x_2$}\\
&&\hspace{100 pt}\Rightarrow f_3(x_1,x_2,x_3)=f_3(x_2,x_1,x_3)\text{ for all }x_1,x_2,x_3.\label{req1}
\end{eqnarray}

Next, in order to preserve positivity of $f_1$ and $f_2$, we need to impose 
\begin{equation}\label{req2}
\text{for all }x_1,x_2: (f_2(x_1,x_2)=0\Rightarrow f_3(x_1,x_2,x_3)=0 \text{ for all }x_3)
\end{equation}
and 
\begin{equation}\label{req3}
(f_2(x_1,x_2)\geq 0\text{ for all }x_1,x_2)\Rightarrow (f_3(x_1,x_2,x_3)\geq 0 \text{ for all }x_1,x_2,x_3).
\end{equation}
The requirement \eqref{req2} implies that there exists a function $h(x_1,x_2,x_3)$ such that $f_3(x_1,x_2,x_3)=h(x_1,x_2,x_3)f_2(x_1,x_2)$. Indeed, if $f(x_1,x_2)\neq 0$, then \begin{equation}\nonumber h(x_1,x_2,x_3)=\frac{f_3(x_1,x_2,x_3)}{f_2(x_1,x_2)}\text{ for all $x_3$.}
\end{equation}
If $f_2(x_1,x_2)= 0$, then $h$ can be defined arbitrarily.
By the method of characteristics, this property guarantees positivity of the solutions to the truncated system \eqref{eq_for_f2} provided that the initial data is positive. 

Finally, in order to have the consistency property $f_1(x)=\int f_2(x,y)dy$ we impose 
\begin{equation}\label{req4}
f_2(x_1,x_2)=\int f_3(x_1,x_3,x_2)dx_3
\end{equation}
The equality \eqref{req4} is equivalent to the statement that if we integrate the equation for $k=2$ from the BBGKY hierarchy with respect to one of the spatial variables, say, $x_2$, we get the equation for $k=1$.

\begin{proposition} There is no such representation \eqref{repr} that all requirements \eqref{req1},\eqref{req2},\eqref{req3} and \eqref{req4} hold true. 
\end{proposition}
\begin{proof}
The proof is by contradiction. The idea is to combine requirements \eqref{req2} and \eqref{req3}: 
\begin{equation}\label{comb34}
f_2(x_1,x_2)=\int h(x_1,x_3,x_2) f_2(x_1,x_3)dx_3
\end{equation} 
and to find such $f_2$ that the LHS of \eqref{comb34} is zero, but the RHS is not zero. 

Assume that \eqref{req1},\eqref{req2},\eqref{req3} and \eqref{req4} hold true. Take $$\Omega=\left\{(x_1,x_2): |x_1-\frac{1}{2}|+|x_2-\frac{1}{2}|<\frac{1}{2}\right\}\backslash\left\{|x_1-\frac{1}{2}|<\frac{1}{4},|x_2-\frac{1}{2}|<\frac{1}{4}\right\}$$
and $f_2(x_1,x_2)=\frac{1}{|\Omega|}\chi_{\Omega}(x_1,x_2)=4\chi_{\Omega}(x_1,x_2)$.  Here $\chi_{\Omega}$ is a characteristic function of domain $\Omega$. Note $f_1(x)>0$ for all $x\in (0,1)\backslash \left\{\frac{1}{4},\frac{3}{4}\right\}$ because of the equality $f_1(x)=\int f_2(x,y)dy$ which holds due to \eqref{req4}.

\begin{figure}
\centerline{\includegraphics[height=5 cm]{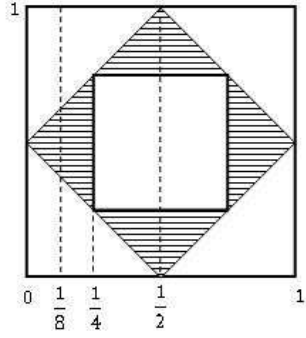}}
\caption{$\Omega$ is shaded domain}
\end{figure}
\noindent The property \eqref{req2} implies the existence of such a function $h(x_1,x_2,x_3)$ that $f_3(x_1,x_2,x_3)=h(x_1,x_2,x_3)f_2(x_1,x_2)$. Thus, from \eqref{req4} we obtain 
\begin{equation} \label{l5}
f_2(x_1,x_2)=\int h(x_1,x_3,x_2)f_2(x_1,x_3)dx_3.
\end{equation}  
Let $(x_1,x_2)\notin \Omega$, then \eqref{l5} implies that 
\begin{equation}
0=\int h(x_1,x_3,x_2)f_2(x_1,x_3)dx_3 =4\int_{x_3:(x_1,x_3)\in \Omega}h(x_1,x_3,x_2)dx_3.
\end{equation} 
Thus
\begin{equation}\label{omega1}\text{$h(x_1,x_3,x_2)=0$, if $(x_1,x_2)\notin \Omega$ and $(x_1,x_3)\in \Omega$.}
\end{equation}
By using the symmetry of $h$ with respect to  first two arguments  we get $h(x_1,x_3,x_2)=h(x_3,x_1,x_2)$ and 
\begin{equation}\label{omega2}
h(x_1,x_3,x_2)\equiv 0\text{ if }(x_2,x_3)\notin \Omega\text{ and }(x_1,x_3)\in \Omega. 
\end{equation}
Finally, calculate $f_1(1/8)$. On the one hand, $f_1(1/8)=\int f_2(1/8,y)dy>0$. On the other hand, 
\begin{equation}\label{intint}
f_1(x_2)=4\int\int_{(x_1,x_3)\in \mathcal{O}}h(x_1,x_3,x_2)\chi_{\Omega}(x_1,x_3)dx_3dx_1.
\end{equation}
where $\mathcal{O}=\left\{(x_1,x_3):h(x_1,x_2,x_3)\chi_{\Omega}(x_1,x_3)\neq 0\right\}$. The domain $\mathcal{O}$ depends on $x_2$. We claim that $\mathcal{O}$ is empty for $x_2=1/8$. Indeed, 
\begin{eqnarray*}
\mathcal{O}&=&\left\{h(x_1,x_3,1/8)\neq 0 \text{ and }\chi_{\Omega}(x_1,x_3)\neq 0\right\}=[\text{defenition of $\chi_{\Omega}$}]\\
&=&\left\{(x_1,x_3)\in \Omega, h(x_1,x_3,1/8)\neq 0\right\}\subset [\text{\eqref{omega1} and \eqref{omega2}}]\\
&\subset&  \left\{(x_1,x_3)\in \Omega, (x_1,1/8)\in \Omega,(x_3, 1/8)\in \Omega\right\}\\
&=&\left\{(x_1,x_3)\in \Omega, x_1\in (3/8,5/8), x_3\in (3/8,5/8)\right\}=\emptyset.
\end{eqnarray*}

Therefore, integral in \eqref{intint} is taken over empty set. Thus, $f_1(1/8)=0$ and we have reached a contradiction.
\end{proof}

\medskip

Instead of using a unique representation ansatz for $f_3$ we use two different,  but similar, representation ansatzes for $f_3$, $f_3=f_3^{(\text{I})}(x_1,x_2,x_3)$ and $f_3=f_3^{(\text{II})}(x_1,x_2,x_3)$, in two different places where $f_3$ appears in the equation $k=2$ such that the key properties are preserved. Namely, the equation $k=2$ is rewritten as follows 
\begin{eqnarray}
\nonumber&&\partial_t f_2 +\frac{\alpha K(0)}{N}\partial_{x_1}f_2 +\frac{\alpha K(0)}{N}\partial_{x_2}f_2\\
\nonumber&&\hspace{50 pt}+ \frac{\alpha}{N}\partial_{x_1} (K(x_2-x_1)f_2)+\frac{\alpha}{N}\partial_{x_2}(K(x_1-x_2)f_2)\\
\nonumber&&\hspace{50 pt} +\alpha\frac{N-2}{N}\partial_{x_1}\left\{\int K(x_3-x_1)f_3^{(\text{I})}(t,x_1,x_2,x_3)dx_3\right\}\\
\label{trunc} &&\hspace{50 pt} +\alpha\frac{N-2}{N}\partial_{x_2} \left\{\int K(x_3-x_2)f_3^{(\text{II})}(t,x_1,x_2,x_3)dx_3\right\}=0,
\end{eqnarray}
where
\begin{equation}\label{ansatz1}
f^{(\text{I})}_3(t,x_1,x_2,x_3)=\left\{\begin{array}{ll}\frac{f_2(t,x_1,x_2)f_2(t,x_1,x_3)}{\int f_2(t,x_1,y)dy},&\int f_2(t,x_1,y)dy>0,\\0,&\int f_2(t,x_1,y)dy=0,\end{array}\right.
\end{equation} and 
\begin{equation}\label{ansatz2}
f^{(\text{II})}_3(t,x_1,x_2,x_3)=\left\{\begin{array}{ll}\frac{f_2(t,x_1,x_2)f_2(t,x_3,x_2)}{\int f_2(t,y,x_2)dy},& \int f_2(t,y,x_2)dy>0,\\0,&\int f_2(t,y,x_2)dy=0.\end{array}\right.
\end{equation}
The four key properties listed below \eqref{eq_for_f2} {are preserved after such truncation}:
\begin{list}{}{}
\item {1.} Symmetry of $f_2(t,x_1,x_2)$ with respect to $x_1$ and $x_2$ (provided that $f_2(0,x_1,x_2)$ is symmetric) follows from symmetry of the equation with respect to $x_1$~and~$x_2$. 
\item {2.} Conservation of mass and positivity follow from the fact that equation \eqref{trunc} can be rewritten as a standard conservation law (see \eqref{trunc_as_cl} below).   
\item {3.}  By integrating \eqref{trunc} over, for example, $x_2$, one obtains an equation for $\int f_2(t,x_1,x_2)dx_2$ which coincides with the equation for $f_1$. By assuming uniqueness we get the consistency property: $f_1(t,x_1)=\int f_2(t,x_1,x_2) dx_2$.
\item {4.} If $f_2(x_1,x_2)=f_1(x_1)f_1(x_2)$, then $f^{(\text{I})}_3=f^{(\text{II})}_3=f_1(x_1)f_1(x_2)f_1(x_3)$. Note that in this case equation \eqref{trunc} {in the limit $N\rightarrow \infty$} is reduced to the Mean Field equation \eqref{vlasov}. 

\end{list}

We conclude this section by giving a physical interpretation of the introduced ansatz. To this end, we will rewrite \eqref{trunc} in a more convenient form.

Substitute \eqref{ansatz1} into the first integral term  in \eqref{trunc} 
$$
\int K(x_3-x_1)f^{\text{(I)}}(t,x_1,x_2,x_3)dx_3.
$$ Then this term is of the form 
\begin{equation}\label{other_for_of_integr_term}
F(t,x_1)f_2(t,x_1,x_2),
\end{equation}
where 
\begin{equation}\label{force} 
F(t,x)=\int K(y-x)f_2(t,x,y)dy/\int f_2(t,x,y)dy
\end{equation}
 (assume $f_2>0$).  Analogously, the second integral term is $
F(t,x_2)f_2(t,x_1,x_2)
$.

Recall that $K(y-x)$ is an interaction kernel and therefore it can be viewed as a force exerted by the particle located at $y$ on a particle located at $x$. Thus, the RHS of \eqref{force} is a total force exerted on particles located at $x$ by all other particles whose location at time $t$ is described by variable $y$ (the density of these particles is the normalized $f_2(t,x,y)$). 

Next, substituting \eqref{force} and \eqref{other_for_of_integr_term} into \eqref{trunc} we obtain a conservation law for $f_2(t,x_2,x_2)$
\begin{eqnarray}
\nonumber&&\partial_t f_2 +\partial_{x_1} \left(\left\{\frac{\alpha K(0)}{N}+\frac{\alpha}{N}K(x_2-x_1)+\alpha\,\frac{(N-2)}{N}\,F(t,x_1)\right\}f_2\right)\\
\label{trunc_as_cl}&&\hspace{50 pt}+\partial_{x_2} \left(\left\{\frac{\alpha K(0)}{N}+\frac{\alpha}{N}K(x_1-x_2)+\alpha\frac{(N-2)}{N}\, F(t,x_2)\right\}f_2\right)=0.
\end{eqnarray}
The first term in curly braces represents self-interaction, the second term represents the force exerted on the particle located at $x_1$ by a particle at $x_2$ and the third term represents the force exerted by the remaining $N-2$ particles on the particle located at $x_1$.


Finally, rewrite the Vlasov equation \eqref{vlasov} (MF equation) with no self-propulsion in the following form 
\begin{equation}\label{vlasov_g}
\partial_t f+\alpha\partial_{x}\left(\mathcal{G}(t,x)f\right)=0, \text{ where }\mathcal{G}(t,x)=\int K(x-y)f(t,y)dy.
\end{equation}
A comparison of \eqref{trunc_as_cl} and \eqref{vlasov_g} shows that equation \eqref{trunc_as_cl} can still be viewed as a mean field approximation; but at a higher order, 
 that is  for $f_2(t,x_1,x_2)$ instead of $f(t,x)=\lim\limits_{N\to\infty}f_1(t,x)$ and with the correct coefficients including order $1/N$ corrections.

\section{Numerical example}\label{numerics}  
The goal of this section  is (i) to test the truncation \eqref{trunc} on a simple 1D example 
and (ii) to describe one way of handling nonlocality and nonlinearity in the numerical resolution. 
Here we use numerical  methods which are explicit, allows for comparison with direct simulations, and {are} not necessarily the most efficient. 
The development of more advanced numerical methods capturing, e.g., 2D, non-smooth kernels or large times, are left for a subsequent work. 

To test the truncation we compare probability distributions (marginals) $f_1$ and $f_2$ obtained by numerical solution of the truncated system \eqref{trunc} with histograms of particles satisfying the original ODE system \eqref{original}. The histograms are built on many realizations of initial particle positions. 

\medskip

\noindent{\it Specific setting.} Numerical simulations are performed  for the one-dimensional problem, $x\in \mathbb R^1$, and periodic boundary conditions with period 1. The interaction kernel $K$ is periodic with period 1 and for $|X_j-X_i|<1/2$ it is given by  $K(X_j-X_i)=e^{-12(X_j-X_i)^2}$. Initially particles are independent: 
\begin{equation}\label{ic}
f_2(0,x_1,x_2)=f_1(0,x_1)f_1(0,x_2),\text{ where }f_1(0,x_1)=.4 \sin 2\pi x_1 + 1. 
\end{equation}
Number of particles is $N= 100$ per one periodic cell $x\in[0,1]$, $\alpha=3$.  

\medskip 

\noindent{\it Description of numerical methods.} In order to solve the PDE \eqref{trunc} we face difficulties that come from the fact that the equation is a non-local non-linear 2D conservation law. For a detailed discussion of difficulties in  numerical solution of non-linear conservation laws and the way to resolve them we refer to \cite{Lev92}. In this example we want to simulate accurately terms {of order $1/N$}, 
since they are the source of correlations. In other words, if one erases these terms in \eqref{trunc}, then the solution of equation \eqref{trunc} with initial conditions \eqref{ic} will be of the form $f_2(t,x_1,x_2)=f_1(t,x_1)f_1(t,x_2)$, i.e., with no correlations. 
 This motivates us to use a second order scheme and we implemented a second order scheme with flux limiters for which we have converging numerical solutions with reasonable spatial and time steps. This method is described below.   

The PDE \eqref{trunc} can be rewritten as follows   
\begin{equation*}
\partial_tf_2+\partial_{x_1}(\mathcal{A}_1 f_2)+\partial_{x_2}(\mathcal{A}_2 f_2)=0,
\end{equation*}
where for $f_2(x_1,x_2)>0$ functions $\mathcal{A}_1$ and $\mathcal{A}_2$ are given by
\begin{equation}
\mathcal{A}_k(t,x_1,x_2)=\frac{\alpha}{N}\sum\limits_{i=1,2}K(x_i-x_k) +\frac{\alpha(N-2)}{N} \frac{\int K(y-x_k)f_2(t,x_k,y)dy}{\int f_2(t,x_k,y)dy}.
\end{equation}

Denote by $f^m_{i,j}$ the approximation for $f_2(t,x_1,x_2)$ with $t=m\text{d}t$, $x_1=i \text{d}x$, $x_2=j \text{d}x$, where $\text{d}t$ and $\text{d}x$ are time and spatial steps, respectively. For given $m$, $i$ and $j$ introduce the following finite difference approximations for $\partial_k\left\{\mathcal{A}_{k}f_2\right\}$, $k=1,2$: 
\begin{eqnarray*}
&& r_{11}:=\frac{A^{m}_{i,j}f^{m}_{i,j}-A^{m}_{i-1,j}f^{m}_{i-1,j}}{\text{d}x},\;\;r_{12}:=\frac{A^{m}_{i+1,j}f^{m}_{i+1,j}-A^{m}_{i,j}f^{m}_{i,j}}{\text{d}x},\\
&& r_{21}:=\frac{A^{m}_{i,j}f^{m}_{i,j}-A^{m}_{i,j-1}f^{m}_{i,j-1}}{\text{d}x},\;\;r_{22}:=\frac{A^{m}_{i,j+1}f^{m}_{i,j+1}-A^{m}_{i,j}f^{m}_{i,j}}{\text{d}x}.
\end{eqnarray*}
Introduce also an auxiliary function (flux limiter) $\phi(r)=\max\left[ 0, 0.5 \min (r, 1.5)\right]$.

The following finite difference scheme is used in the numerical solution of PDE \eqref{trunc}:
\begin{equation}\label{marg_calc}
f^{m+1}_{i,j}=f^{m}_{i,j} + \frac{\text{d}t}{\text{d}x}\left[\sum\limits_{k=1,2}\left\{r_{k1}+\phi \left(\frac{r_{k1}}{r_{k2}}\right)(r_{k2}-r_{k1})\right\}\right].
\end{equation}

In order to compute the two-particle distribution for $t>0$ directly from the system of ODEs \eqref{original} we consider $R=5\cdot 10^{5}$ realizations of $N=100$ particles initially identically distributed with probability distribution function $f(x)=0.4\sin 2 \pi x+1$. Denote by $X_{i}^{(r)}(t)$ the position of the $i$th particle, $i=1,..,N$ in the $r$th realization, $r=1,..,R$ at time $t$. For each $r=1,..,R$ the positions $\left\{X_i^{(r)}(t)\right\}_{i=1,..,N}$, $t>0$, are found as the solution of the ODE system \eqref{original} by the explicit Euler method of the first order with the time step $\Delta t=0.01$.     

We compute the following histogram which approximates the probability of that the first particle is in the interval $\Delta_j=[jh,(j+1)h)$ at time $t$:
\begin{equation}\label{hysto_f1}
\tilde{f}_1(t,\Delta_j)=\frac{1}{Rh}\#\left\{X_1^{(r)}(t)\in \Delta_j, \;r=1,..,R\right\}.
\end{equation}
Here $h=0.05$ is the size of a histogram bin. 

Histogram $\tilde{f}_2$ which approximates the two-particle distribution can be computed as follows
\begin{equation}\label{hysto_f2}
\tilde{f}_2(t,\Delta_i,\Delta_j)=\frac{1}{Rh^2}\#\left\{\left(X_1^{(r)}(t),X_2^{(r)}(t)\right)\in \Delta_i\times\Delta_j, r=1,..,R\right\}.
\end{equation}

\medskip

Thus, in numerical simulations  our intention is to compare $f_1=\int f_2 dx$ and $f_2$ calculated by \eqref{marg_calc} with histograms $\tilde{f}_1$ and $\tilde{f}_2$ calculated by \eqref{hysto_f1} and \eqref{hysto_f2}.

\medskip 

Simulations were performed on a machine with 3.06 Ghz Intel core CPU 8 GB of RAM.
Numerical solution of \eqref{trunc} for $t=1$ for $dx=.0025$ and $dx/dt=200$ takes approximately 32 hours. Numerical solution of \eqref{original} on $R=10^5$ realization, $t=1$ and time step $\Delta t=1/50$ takes approximately 83 hours. Besides the long time of computations direct simulations face another difficulty which is the large amount of data 
that creates technical difficulties in data movement, its analysis and visualization. Also note that the cost of direct simulations would increase much faster with $N$ than the cost of our approach.


\noindent{\it Results of numerical simulations.} Plots in Fig. \ref{fig:comppf1} show that marginal $f_1$ {is} close to histogram $\tilde{f}_1$. 

In order to visualize comparisons between $f_2$ and $\tilde{f}_2$ we plot these functions integrated over domain $\mathcal{B}=\left\{(x_1,x_2): 0\leq x_1,x_2 \leq 1/2\right\}$:
\begin{equation*}
\text{Qmarg}:=\int_{\mathcal{B}}f_2 dx_1dx_2,\;\;\;\text{Qhist}:=\int_{\mathcal{B}}\tilde{f}_2 dx_1 dx_2=h^2\sum_{i,j\leq 1/(2h)}\tilde{f}_2(t,\Delta_i,\Delta_j).
\end{equation*}
Plot \ref{Btest} shows that quantities Qmarg defined by marginal $f_2$ and Qhist defined by histogram $\tilde{f}_2$ seem to be in very good agreement. 

Notice that the quarter cube $\mathcal{B}$ was chosen arbitrarily (and in particular there is no conservation of mass on $\mathcal{B}$, unlike conservation of mass for the entire cell $[0,1]^2$).  The agreement between Qmarg and Qhist suggests that the integrals of $f_2$ and $\tilde f_2$ over any subdomain of the cell $[0,1]^2$ would similarly be close. 
The apparent periodicity in time is due to the choice of periodic boundary conditions.

Finally, we computed correlations using marginals $f_1$ and $f_2$. For the marginal approach (i.e., solution of \eqref{trunc} and $f_1=\int f_2 dx_2$) correlations are defined as follows 
\begin{eqnarray*}
c(t):=\int\int |f_2(t,x_1,x_2)-f_1(t,x_1)f_1(t,x_2)|dx_1 dx_2.
\end{eqnarray*}

In direct simulations (i.e., solution of ODE \eqref{original} for many random realizations of initial conditions) correlations are defined in a similar way to the above formula with histograms in place of distributions: 
\begin{eqnarray*}
\tilde{c}(t):&=&\int\int |\tilde{f}_2(t,x_1,x_2)-\tilde{f}_1(t,x_1)\tilde{f}_1(t,x_2)|dx_1 dx_2\\
&=&h^2\sum\limits_{i,j}|\tilde{f}_2(t,\Delta_i,\Delta_j)-\tilde{f}_1(t,\Delta_i)\tilde{f}_1(t,\Delta_j)|.
\end{eqnarray*}

As it is seen on Fig. \ref{correls}, plots for correlations computed on marginals and in direct simulations for $R=5\times 10^5$ have similar qualitative behavior and order of magnitude.  The value of correlations  is a small number and thus its computation requires high accuracy to reduce the error to an order less than that of the correlations. In direct simulations, this requires a large number of realizations which make the computations unreasonably long, in contrast to the marginal approach where the computation of correlations is much faster.

Note that correlations observed in this numerical example are not large (in comparison with the maximal possible value of correlations $c_{\text{max}}=2$). In order to observe large correlations (e.g., $\sim 0.1$) we need to solve \eqref{trunc} for large times which is very costly. Moreover, {it is delicate to predict the time when correlations will reach some fixed, large value}. This question is left for subsequent works. Nevertheless, {relatively} small correlations for times of order $1$ may be enough for the solution of the original BBGKY hierarchy to be essentially different from the one obtained by the Mean Field approach. In that case our approach with \eqref{trunc} would still capture the correct solution in contrast to Mean Field. 

\medskip

\noindent{\it Convergence of numerical methods.} Here we show that the convergence of numerical methods used in this section. 

First, consider the calculations of marginals $f_1$ and $f_2$. Comparisons of numerical simulations for various spatial and time steps for $t=1$, $t=2$ and $t=3$ are presented on Figures \ref{f1sec12} and \ref{f1sec3}. 
Convergence of the numerical method in computing $\int_{\mathcal{B}} f_2 dx$ and correlations $c(t)$ is observed on plots in Figure \ref{f1correl}.

Next, look at the calculations for histograms $\tilde{f}_1$ and $\tilde{f}_2$. Plots on Figures \ref{hf1sec12} and \ref{hf1sec3} illustrate convergence of the method for histogram $\tilde{f}_1$ at times $t=1$, $t=2$ and $t=3$, and for histogram $\tilde{f}_2$ summed over the set $\mathcal{B}$. {Several time steps $\Delta t$ are considered}: $\Delta t=0.02$, $\Delta t=0.01$, $\Delta t=0.001$. The number of realizations, $R=10^5$,  is chosen for the width of bin $h=0.02$. It is seen on  Figures \ref{hf1sec12} and \ref{hf1sec3} that such a number of realizations $R$ seems to be enough to have converging numerical solutions for $\tilde{f}_1$ and $\int_{\mathcal{B}}\tilde{f}_2$. To compute correlations $\tilde{c}(t)$ more realizations {would be} needed and plots in Figure \ref{hcorr_conv} show that to estimate $\tilde{c}(t)$ we need {more than} $R=5\cdot 10^5$ realizations with $\Delta t=0.002$.   

\medskip

Numerical simulations presented above show that PDE system \eqref{trunc}  not only preserves the qualitative properties of the probability distribution functions (like positivity, consistency, propagation of chaos, etc.), but also may serve for the study of saturation of correlations in such many particle systems. Correlations play an important role, e.g., in the description of collective motion (see, e.g., \cite{SokAra12}, where transition from individual to collective state is described via correlations).

\begin{figure}[t]
\includegraphics[width=7 cm]{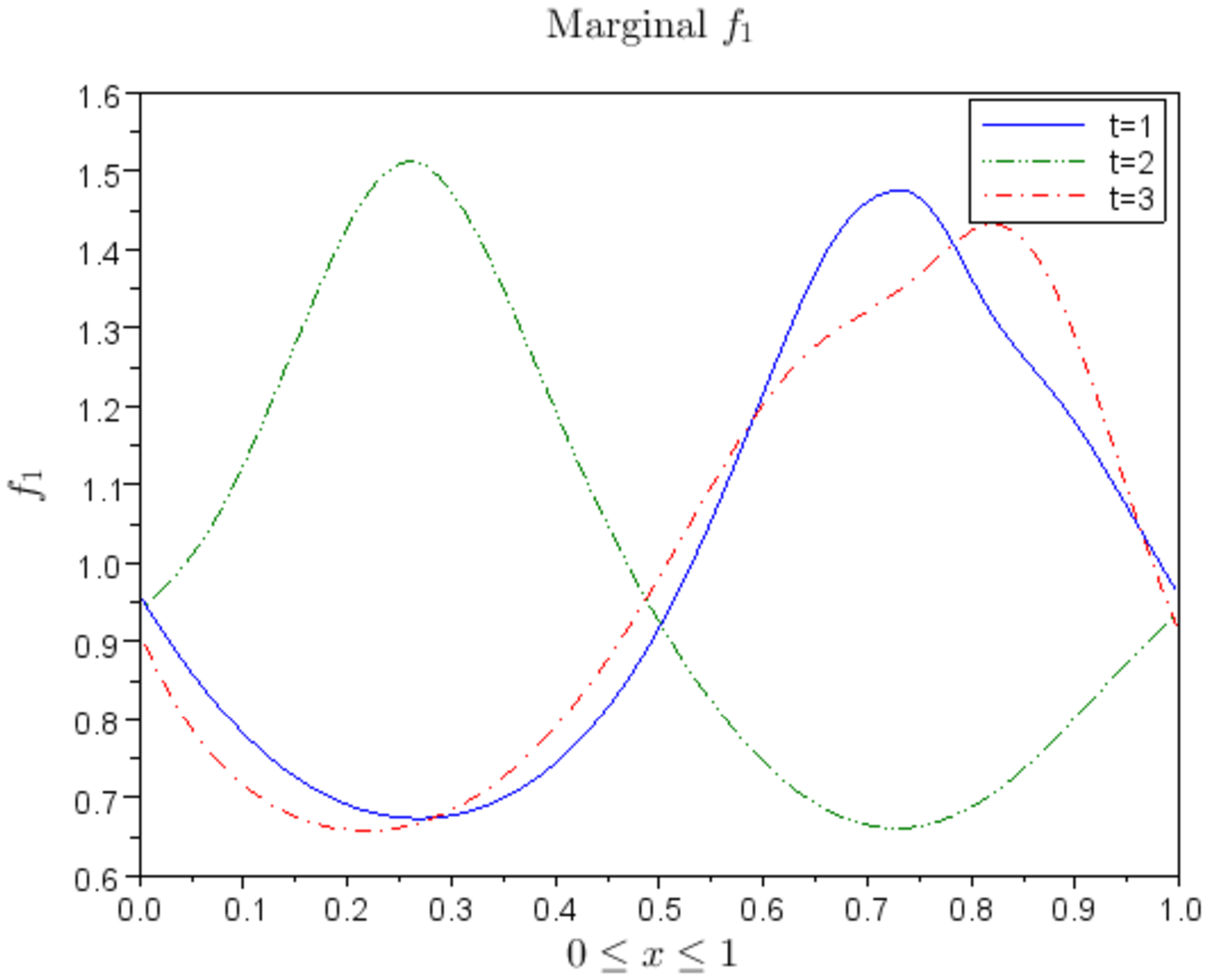}
\includegraphics[width=7 cm]{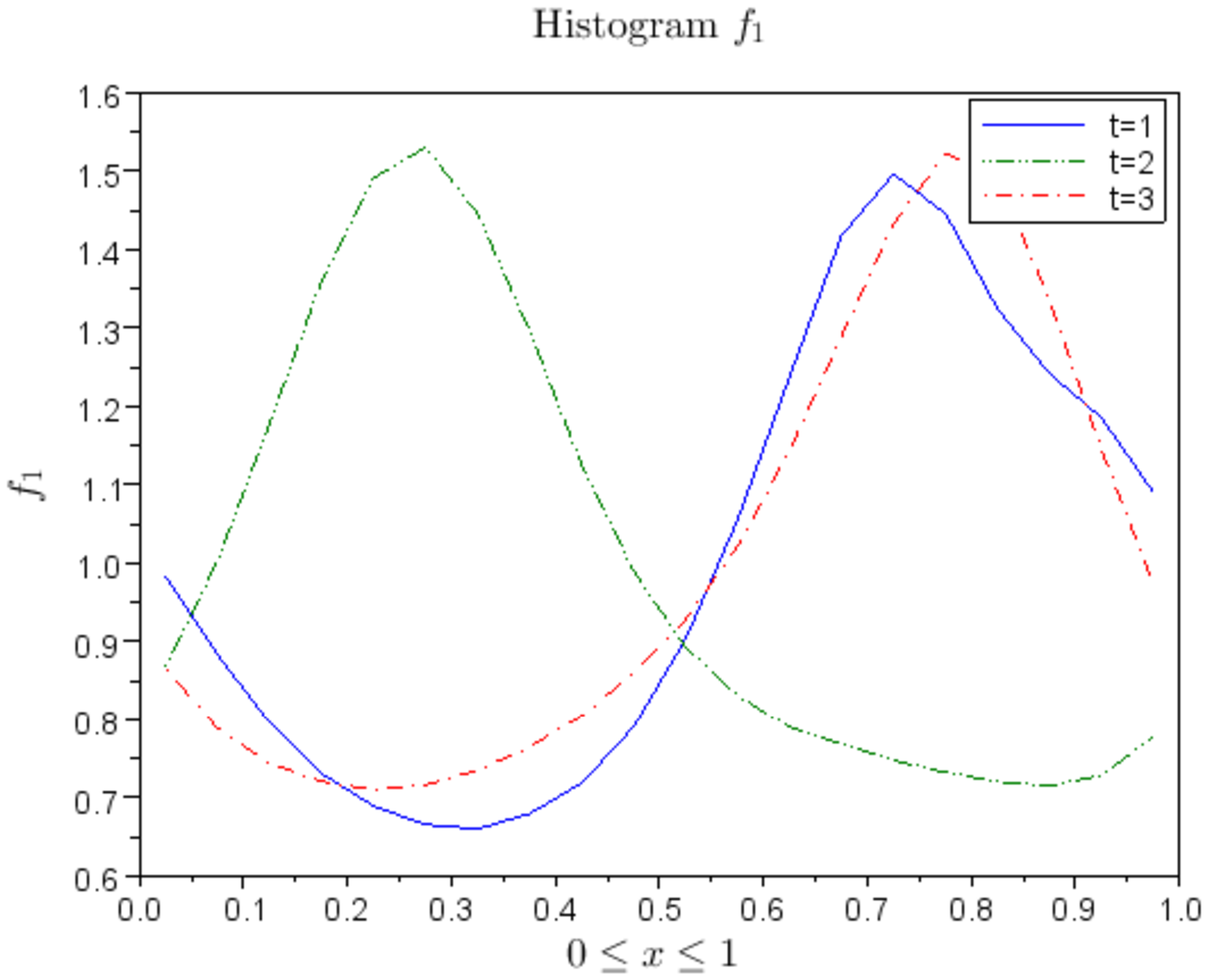}
\caption{Left: Marginal $f_1$ with $dx=0.0025$ and $dt=dx/200$; Right: Histogram $\tilde{f}_1$ for $h=0.05$ and $dt=0.001$.}
\label{fig:comppf1}
\end{figure}

\begin{figure}[t]
\centerline{\includegraphics[width=8 cm]{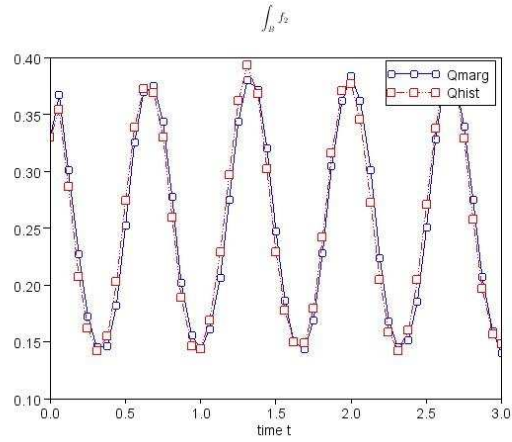}}
\caption{Comparison between {\rm Qmarg} and {\rm Qhist}. }
\label{Btest}
\end{figure}

\begin{figure}[t]
\centerline{\includegraphics[width=8 cm]{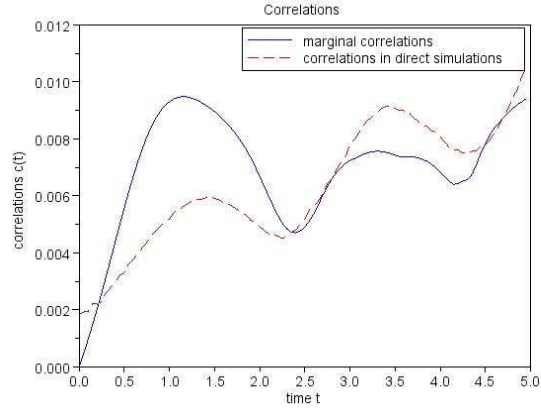}}
\caption{Comparisons of correlations computed by marginal approach \eqref{trunc} and direct simulations of the system \eqref{original}.}
\label{correls}
\end{figure}

\begin{figure}[t]
\centerline{\includegraphics[width=7 cm]{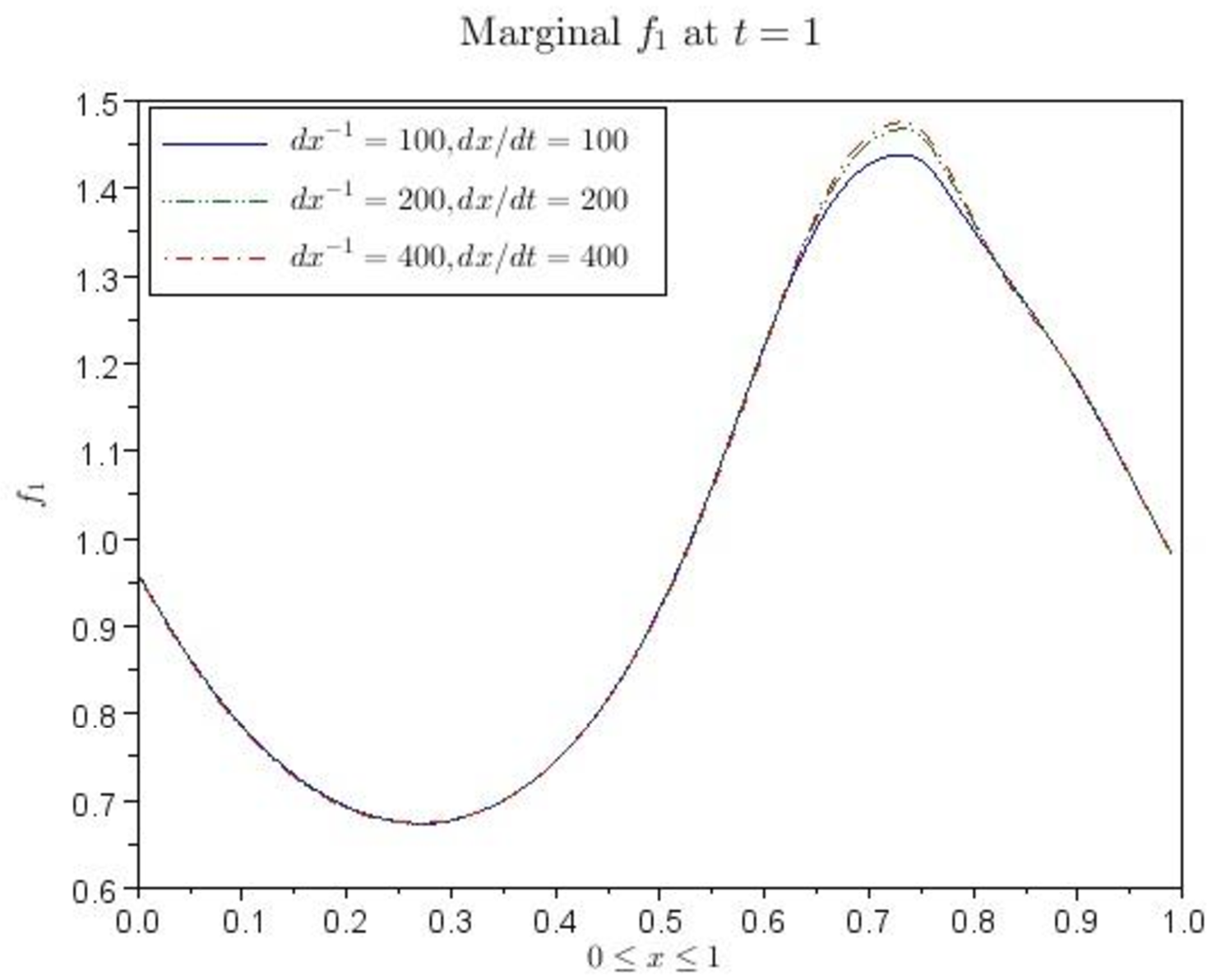} \includegraphics[width=7cm]{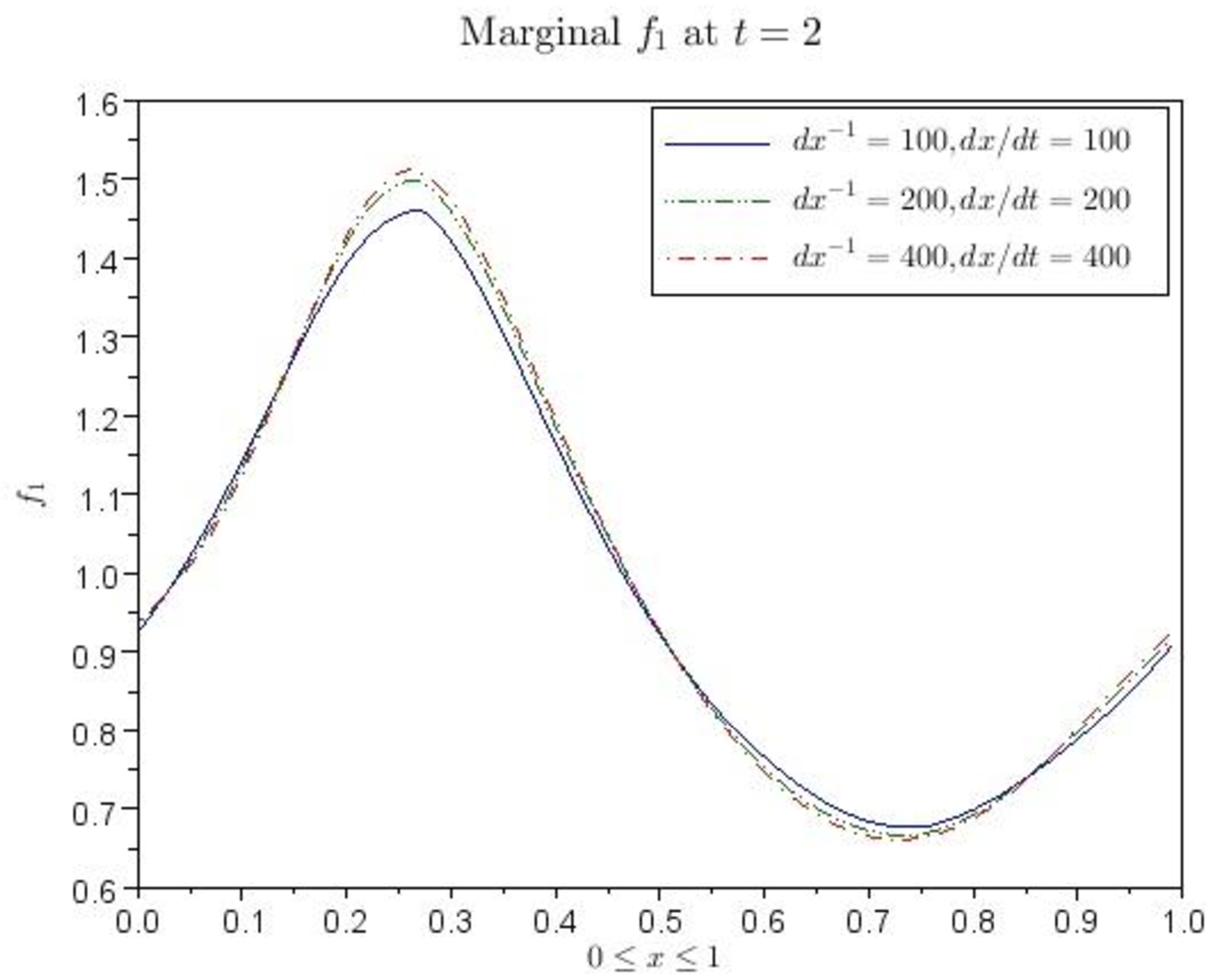}}
\caption{Left: plots of marginal $f_1$ at $t=1$ and Right: $t=2$}
\label{f1sec12}
\end{figure}

\begin{figure}[t]
\centerline{\includegraphics[width=7 cm]{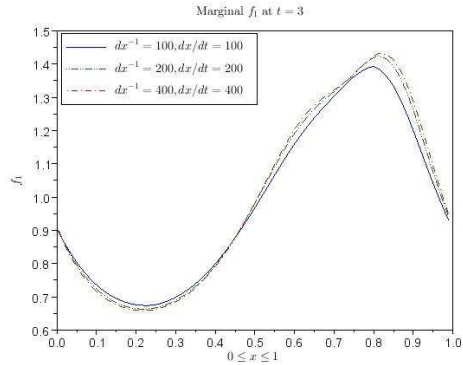}}
\caption{Plots of marginal $f_1$ at $t=3$ for $0\leq x \leq 1$;}
\label{f1sec3}
\end{figure}

\begin{figure}[t]
\centerline{\includegraphics[width= 7 cm]{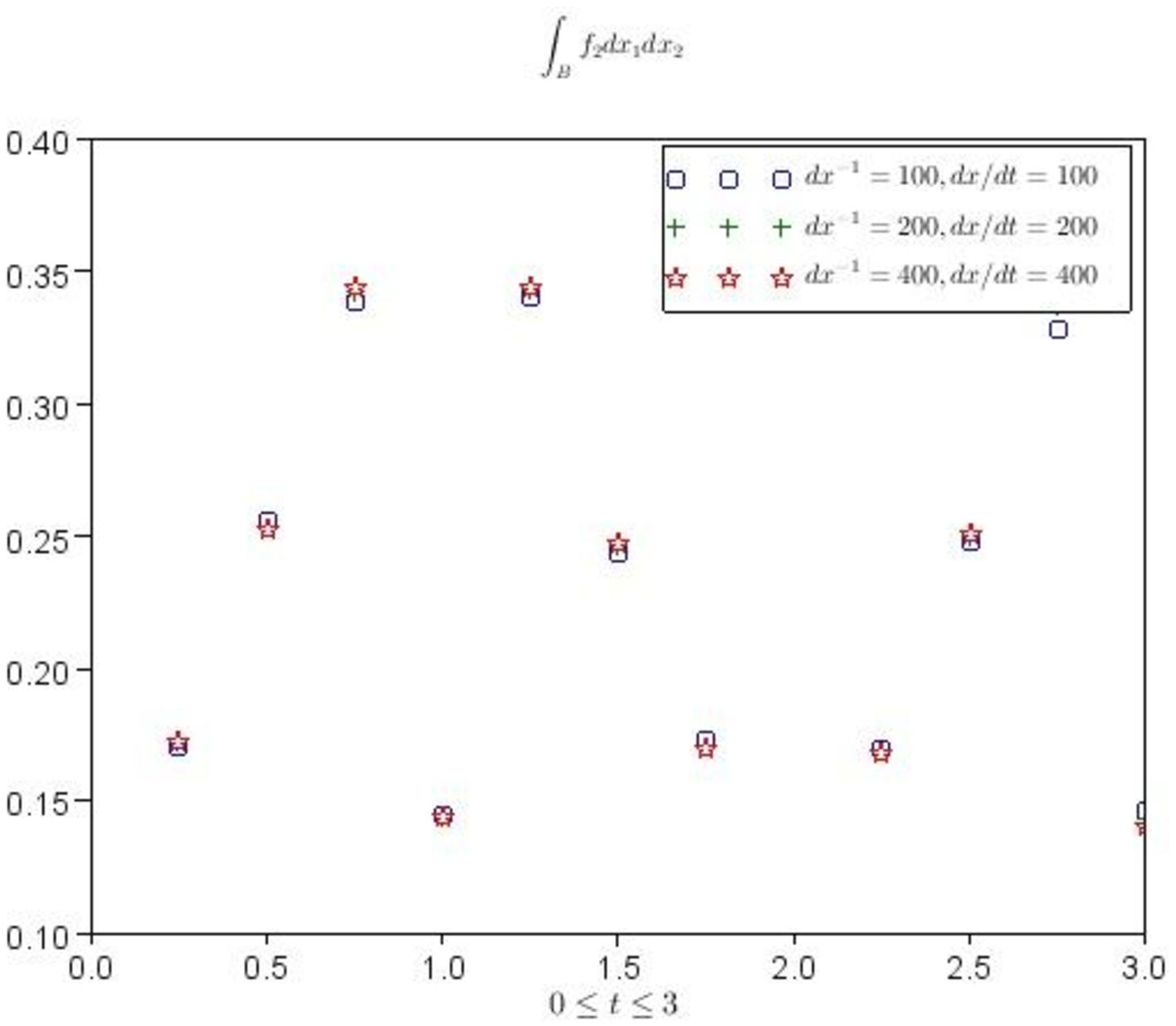} \includegraphics[width=7 cm]{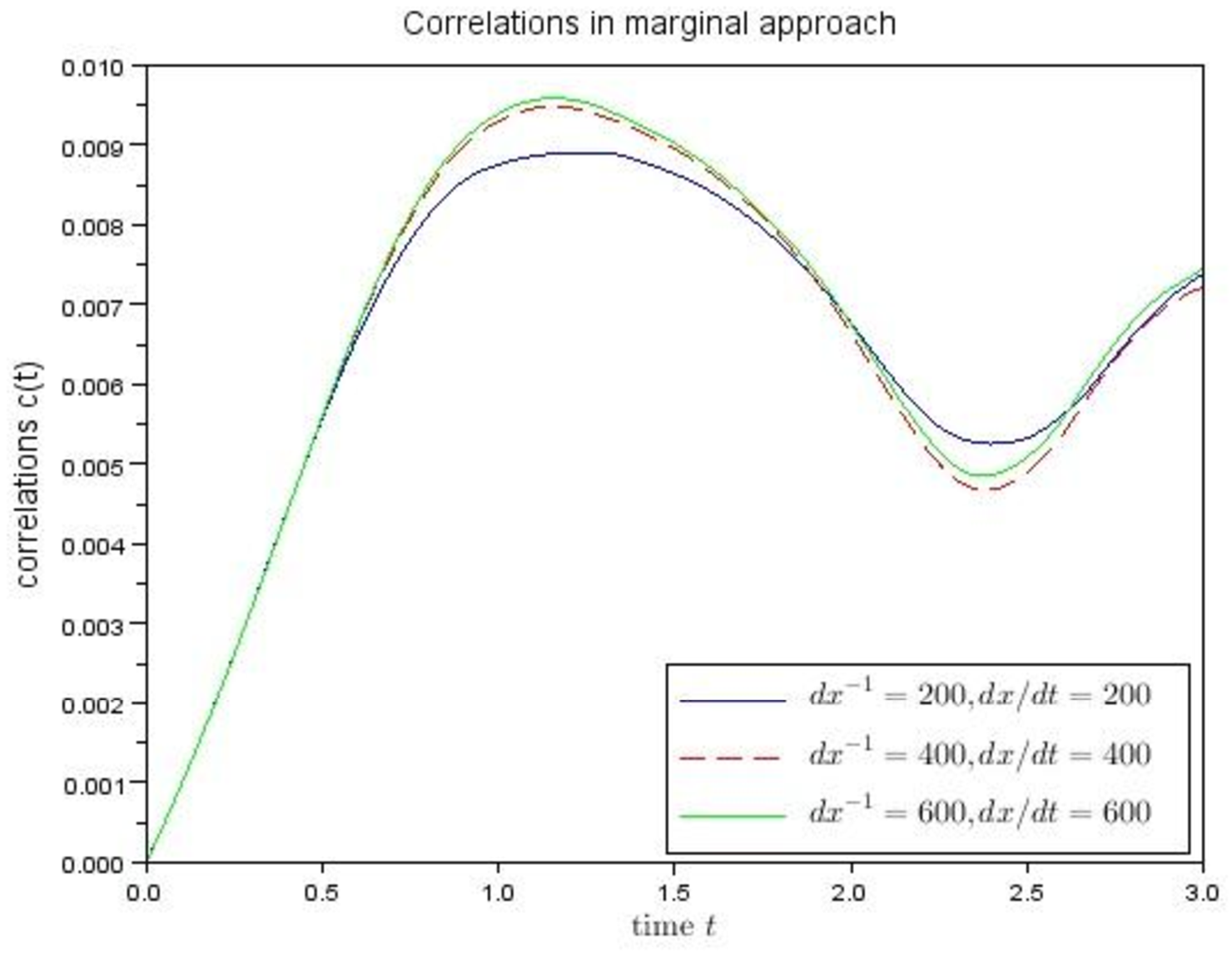}}
\caption{Left: marginal $f_2$ integrated over $B$; Right: correlations $c(t)$}
\label{f1correl}
\end{figure}

\begin{figure}[t]
\centerline{\includegraphics[width= 7 cm]{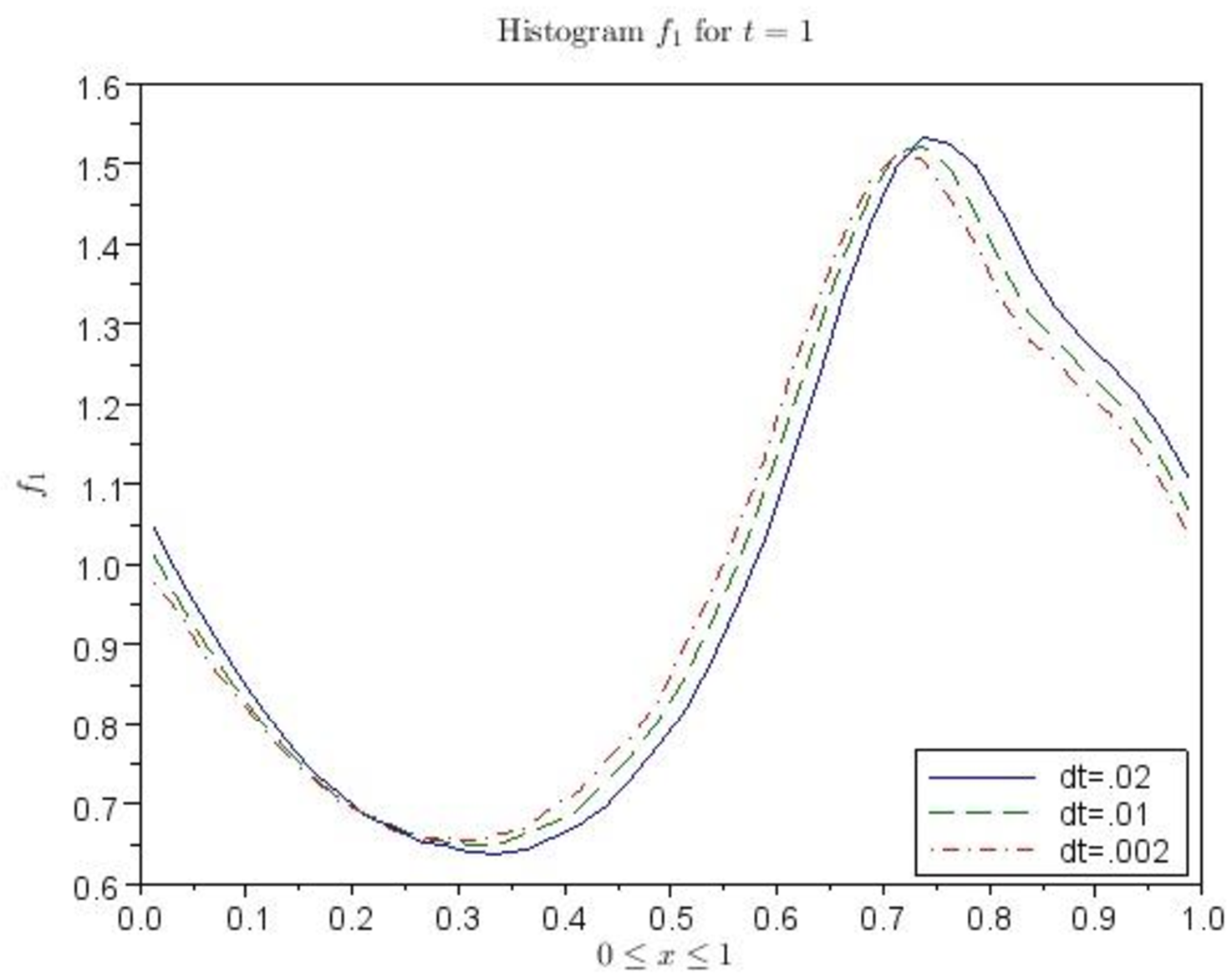} \includegraphics[width=7 cm]{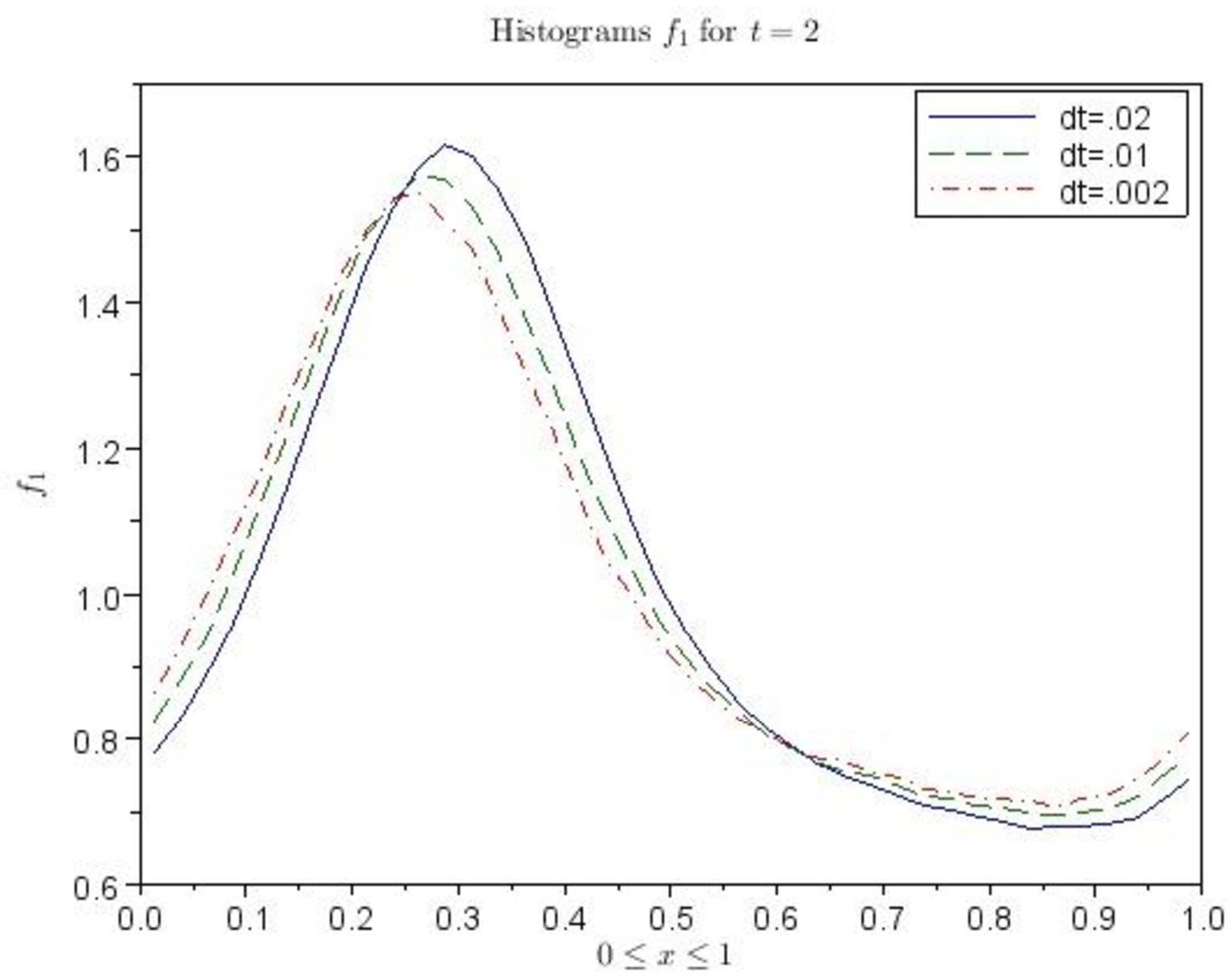}}
\caption{Left: plots of histogram $\tilde{f}_1$ at $t=1$; Right:  plots of histogram $\tilde{f}_1$ at $t=2$}
\label{hf1sec12}
\end{figure}

\begin{figure}[t]
\centerline{\includegraphics[width= 7 cm]{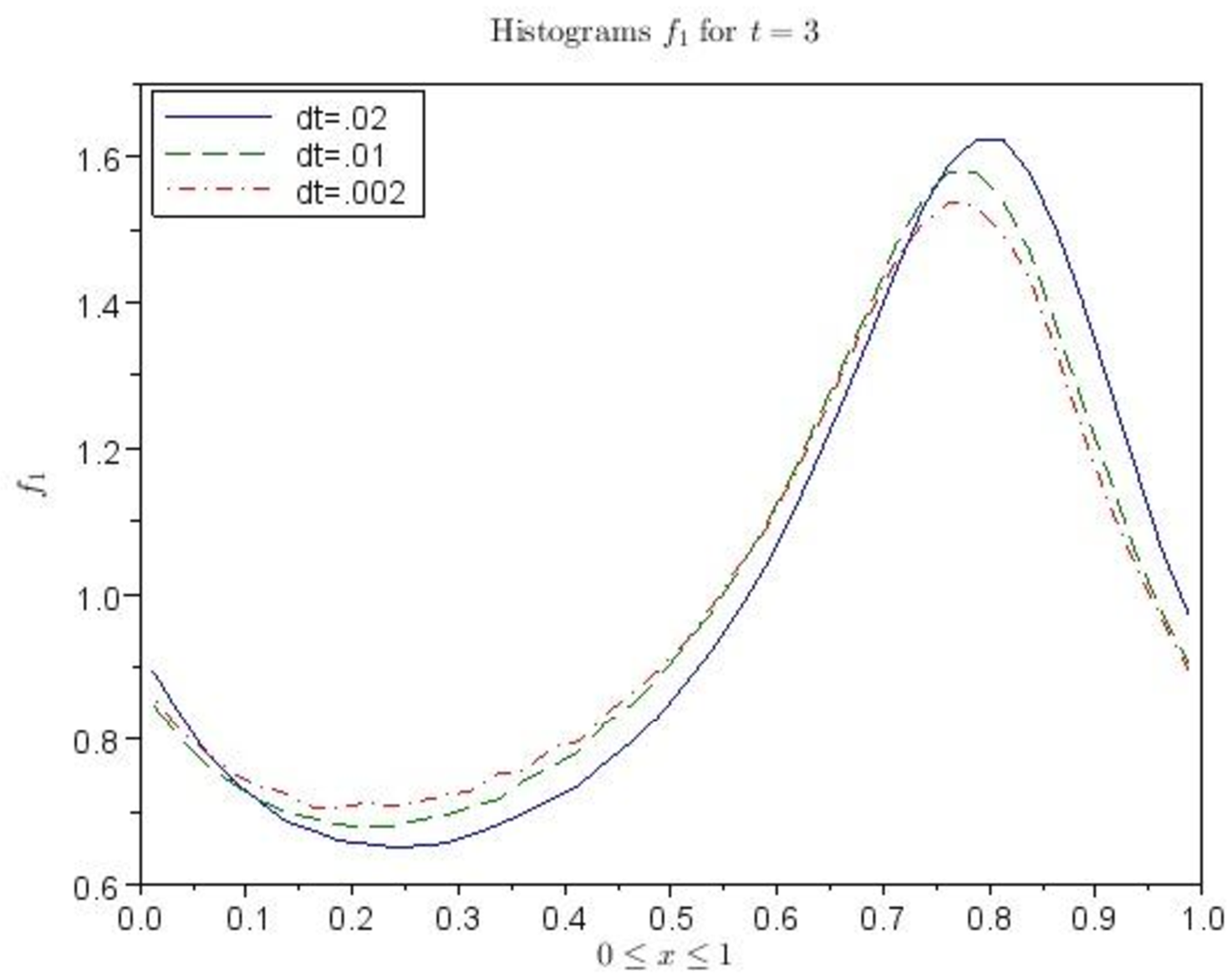} \includegraphics[width=7 cm]{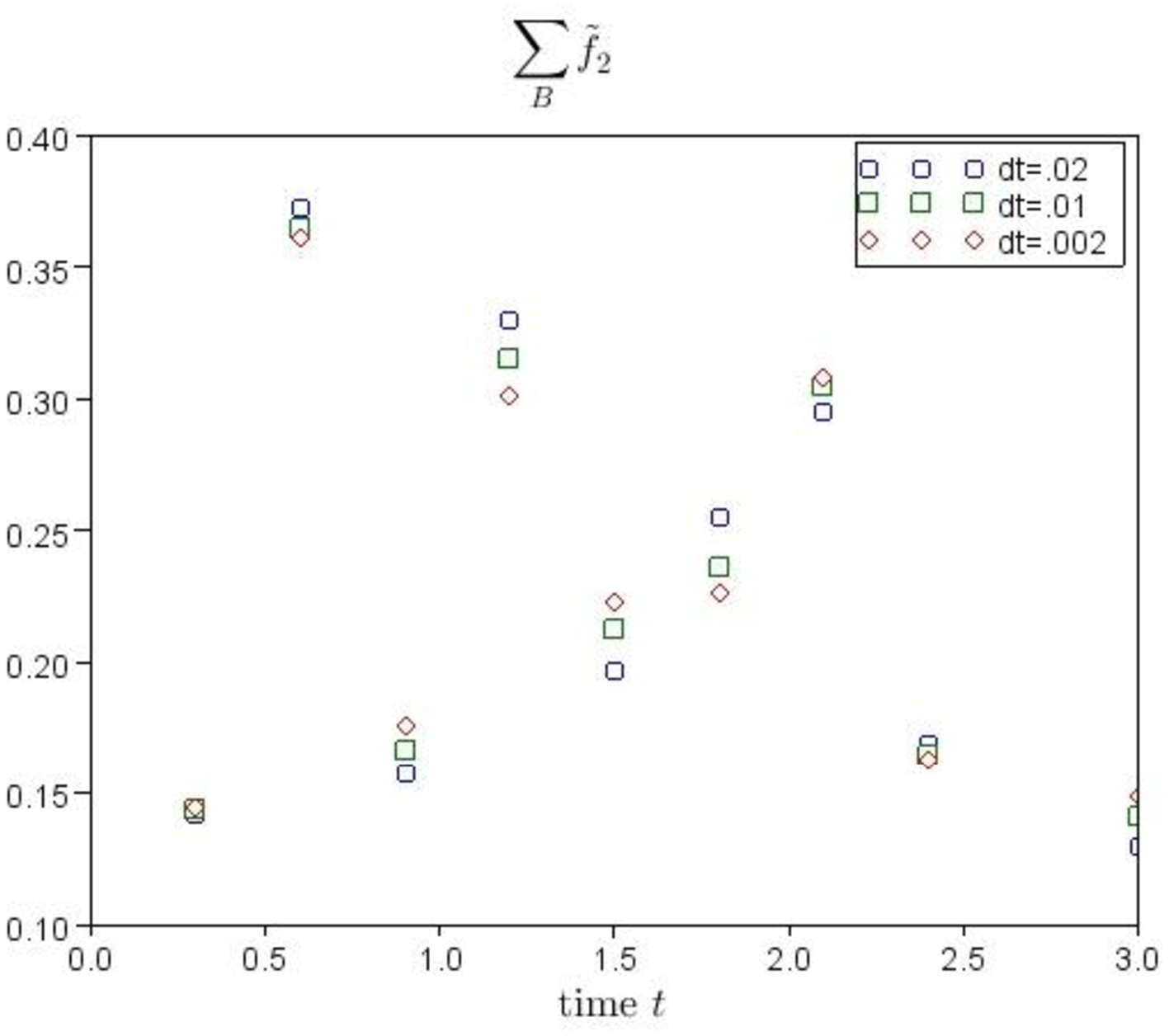}}
\caption{Left: plots of histogram $\tilde{f}_1$ at $t=3$; Right:  plots of histogram $\tilde{f}_2$ summed over set $\mathcal{B}$}
\label{hf1sec3}
\end{figure}

\begin{figure}[t]
\centerline{\includegraphics[width= 7 cm]{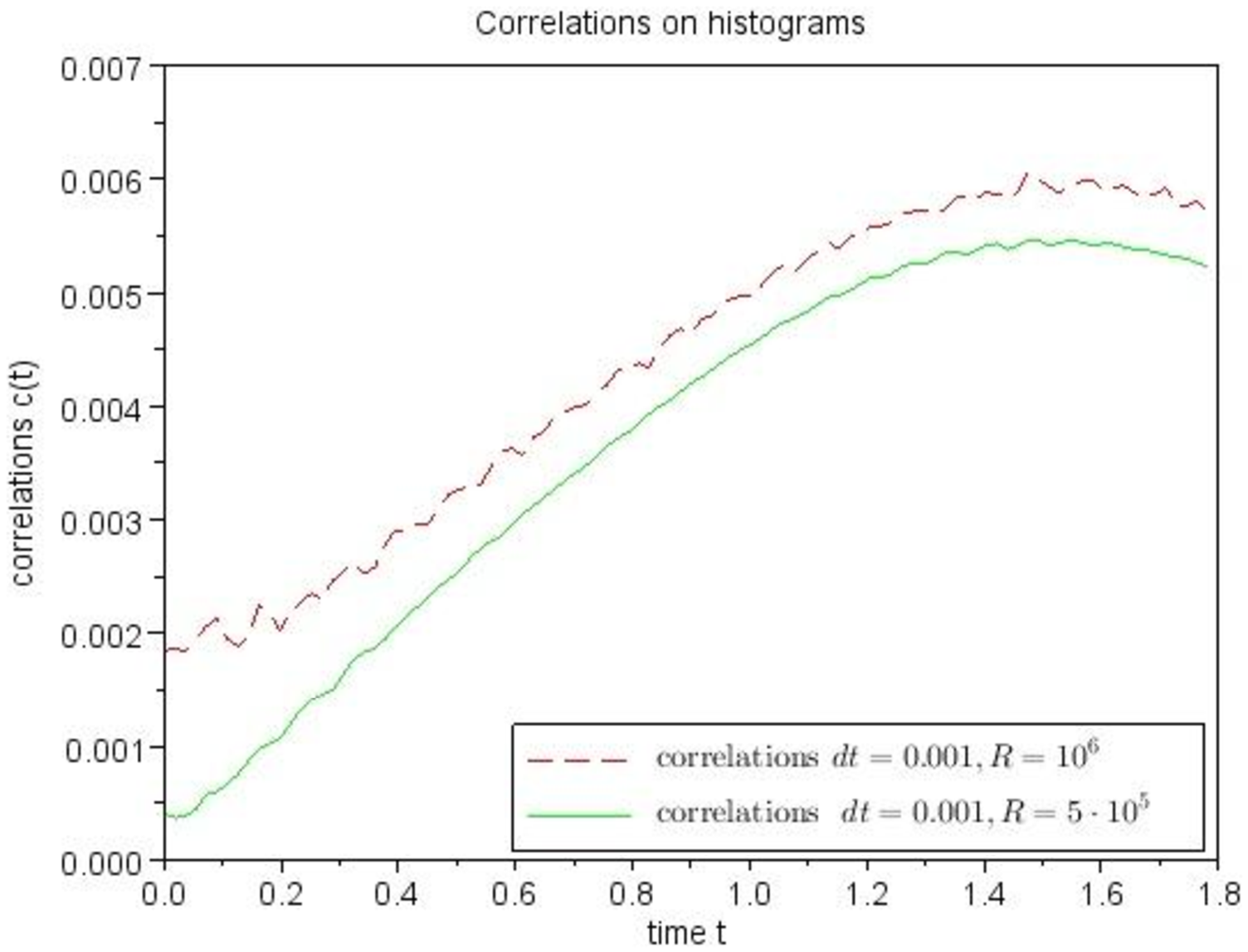}}
\caption{Correlations computed on histograms}
\label{hcorr_conv}
\end{figure}

\medskip 

\section{Conclusions} We developed a numerical approach which allows for study correlations in the evolution of many particle systems  with random initial conditions. This approach is implemented in a simple 1D settings (toy model). The complexity of solving PDE \eqref{trunc} only slowly grows as $N$ goes to infinity. In other words, the dependence of the complexity on $N$ in our approach is more 'innocuous', in sharp contrast with direct simulations when the complexity drastically increases as $N$ grows. 

We believe that this approach can be successfully applied to problems in biology, physics and economics.    
%

\bibliographystyle{ieeetr}
 \bibliography{correlations}
 
\section*{Acknowledgments}
 The work of Leonid Berlyand and Mykhailo Potomkin was supported by DOE grant DE-FG- 0208ER25862. The work of Pierre-Emmanuel Jabin was partially supported by NSF grant DMS-1312142. LB and MP wish to thank V. Rybalko for his comments and suggestions which helped to improve the manuscript.

%
%

\appendix
\section{Appendix: Wasserstein distances}
Wasserstein distance or Monge-Kantarovich-Wasserstein (MKW) quantifies the difference between two given measures. Heuristically, a measure can be viewed as a pile of sand. The MKW distance between two such piles is an optimal work of transfering one pile into another.   

Given two measures $\mu_1$ and $\mu_2$ in $\Pi^1(D)$, one may define the set of transference plans between $\mu_1$ and $\mu_2$ as the set  $\mathcal{T}(\mu_1,\mu_2)$ of measures  $\pi\in \Pi^1(D\times D)$ s.t. 
\[
\mu_1(x)=\int_D \pi(x,dy),\quad \mu_2(y)=\int_D \pi(dx,y).
\]
The $p$ MKW distance $W_p(\mu_1,\mu_2)$ between $\mu_1$ and $\mu_2$ is given by
\[
W_p^p(\mu_1,\mu_2)=\inf_{\pi \in \mathcal{T}(\mu_1,\mu_2)} \int_{D^2} |x-y|^p\,\pi(dx,dy).
\]
If $D$ is the torus, then $|x-y|$ is replaced by the corresponding distance (in general in a manifold, it would be the geodesic distance).

For measures with bounded moments, the MKW distances metrize the weak-* topology. Moreover, on bounded domains the $W_1$ distance is essentially  equivalent to the negative Sobolev norm $W^{-1,1}$.The $p$-MKW distances play an important role for particle systems as the $p$-MKW distance between two empirical measures is typically comparable with the $p$ distance between the two vectors of positions, that is
\[
W_p^p\left(\frac{1}{N}\sum_i \delta_{x_i},\;\frac{1}{N}\sum_i \delta_{y_i}\right)\sim \frac{1}{N} \sum_i |x_i-y_i|^p, 
\]
up to a permutation of indices on the $y_i$.

\end{document}